\documentclass[letterpaper, 10pt, conference]{ieeeconf}
\IEEEoverridecommandlockouts
\usepackage{graphicx} 
\usepackage{epsfig} 
\usepackage{amsmath} 
\usepackage{amssymb}  
\usepackage[noadjust]{cite}
\usepackage{bm}
\usepackage{subfigure}
\usepackage{acronym}
\usepackage{paralist}
\usepackage{float}
\usepackage{epstopdf}
\usepackage{algorithm}
\usepackage[noend]{algpseudocode}
\usepackage{mathtools}
\usepackage{multicol} 
\usepackage{url}
\usepackage{geometry}

\makeatletter
\def\BState{\State\hskip-\ALG@thistlm}
\makeatother

\graphicspath{ {SemiCircPics/} }

\newtheorem{theorem}{Theorem}
\newtheorem{lemma}{Lemma}
\newtheorem{cor}{Corollary}

\newtheorem{define}{Definition}

\DeclareMathOperator{\N}{\mathcal{N}}

\newcommand{\D}{\mathcal{D}}
\newcommand{\V}{\mathcal{V}}
\newcommand{\Ed}{\mathcal{E}_d}
\newcommand{\Eg}{\mathcal{E}_g}

\newcommand{\E}{\mathcal{E}}

\newcommand{\K}[1]{K_{#1}}

\DeclarePairedDelimiter{\ceil}{\lceil}{\rceil}
\DeclarePairedDelimiter{\floor}{\lfloor}{\rfloor}

\begin{document}

\title{\LARGE \bf $r$-Robustness and $(r,s)$-Robustness of Circulant Graphs}

\author{James Usevitch and Dimitra Panagou
\thanks{James Usevitch is with the Department of Aerospace Engineering, University of Michigan, Ann Arbor; \texttt{usevitch@umich.edu}.
Dimitra Panagou is with the Department of Aerospace Engineering, University of Michigan, Ann Arbor; \texttt{dpanagou@umich.edu}.
The authors would like to acknowledge the support of the Automotive Research Center (ARC) in accordance with Cooperative Agreement W56HZV-14-2-0001 U.S. Army TARDEC in Warren, MI.}
}


\maketitle
\thispagestyle{empty}
\pagestyle{empty}

\begin{abstract}
There has been recent growing interest in graph theoretical properties known as $r$- and $(r,s)$-robustness. These properties serve as sufficient conditions guaranteeing the success of certain consensus algorithms in networks with misbehaving agents present. Due to the complexity of determining the robustness for an arbitrary graph, several methods have previously been proposed for identifying the robustness of specific classes of graphs or constructing graphs with specified robustness levels. The majority of such approaches have focused on undirected graphs. In this paper we identify a class of scalable directed graphs whose edge set is determined by a parameter $k$ and prove that the robustness of these graphs is also determined by $k$. We support our results through computer simulations.
\end{abstract}

\IEEEpeerreviewmaketitle

\section{Introduction}
\label{intro}

In recent years there has been a growing amount of attention dedicated to the problem of consensus of a network in the presence of misbehaving agents. The interest in this problem stems from the seminal works \cite{Pease1980reaching} and \cite{Lamport1982byzantine}, in which the problem of reliable agents coming to an agreement in the presence of untrustworthy agents is discussed. 


In \cite{Kieckhafer1994reaching}, a family of algorithms called \emph{Mean-Subsequence-Reduced} (MSR) was introduced which allow normal agents to reach agreement in the presence of faulty or misbehaving agents. In \cite{Leblanc2011consensus} and \cite{Leblanc2012low}, elements of MSR algorithms were used to create a continuous time algorithm called the \emph{Adversarially Robust Consensus Protocol} (ARC-P), which allows normal agents to achieve consensus in the presence of misbehaving agents if the structure of the network satisfies certain conditions and if the total number of misbehaving agents is bounded. The papers \cite{Leblanc2012low} and \cite{Zhang2012robustness} demonstrated that notions traditionally used to describe networks in graph theory (e.g. connectivity, degree) are insufficient to describe the conditions under which algorithms using purely local information can guarantee successful consensus of normal agents.

In \cite{Zhang2012robustness} an algorithm for discrete time systems called \emph{Weighted Mean-Subsequence-Reduced} (W-MSR) was introduced that built upon the concept of MSR algorithms. To describe the conditions under which the W-MSR algorithm can guarantee consensus of normal agents, the authors introduced the concept of \emph{r-robustness}. It was proven that if at most $F$ nodes in the local neighborhood of any normal node are malicious, a sufficient condition for normal nodes to achieve consensus using the W-MSR algorithm is the network being $(2F+1)$-robust. This work was continued in \cite{LeBlanc2012} by introducing the concept of \emph{(r,s)-robustness} as a necessary and sufficient condition for consensus when the total number of malicious adversaries in the network was bounded by a finite constant. The work \cite{Vaidya2012} proved conditions for consensus in a network with Byzantine adversaries. These results are summarized in \cite{LeBlanc_2013_Res} and extended to an adversarial model where up to a certain fraction of each node's neighborhood might be malicious or Byzantine. A continuous-time algorithm with similar necessary and sufficient conditions for consensus based upon r-robustness was presented in \cite{LeBlanc_2013_Res_Continuous}. Additional publications based upon the concepts of $r$- and $(r,s)$-robustness have presented results involving double-integrator dynamics (\cite{Dibaji2014,Dibaji2015,Dibaji2017resilient}), quantized communication (\cite{Wu2016,Dibaji2016resilient,Dibaji2016resilientDelayed}), distributed optimization (\cite{Sundaram2015,Sundaram2016}), synchronization (\cite{LeBlanc2017}), and results dealing with conditions such as asynchronous updates and delays (\cite{LeBlanc2012a,Dibaji2016resilientDelayed,Dibaji2015}), and time-varying networks (\cite{Salda2017})


Vital to the success of the W-MSR algorithm and other algorithms operating on the assumption of $r$-robustness or $(r,s)$-robustness are the assumptions made about the characteristics of the communication topology of the network. In particular, such algorithms often guarantee consensus for a bounded number of adversaries only if the network can be shown to satisfy a certain level of $r$-robustness or $(r,s)$-robustness. This level of robustness is directly related to the upper bound of the number of misbehaving agents that the network can tolerate. From this, it is plain that knowing the robustness of a given graph is highly desirable when working with these algorithms. However, determining the robustness of an arbitrary network is an NP-hard problem \cite{Leblanc2013algorithms}. In \cite{Zhang2015a} it was further specified that determining whether an arbitrary graph satisfied a specified  level of $r$-robustness is a coNP-complete problem. No efficient algorithm currently exists for determining the robustness of an arbitrary graph.

This problem of robustness determination has been approached from several angles. One approach has been to study specific classes of graphs and demonstrate laws that determine their robustness levels. Some examples of graphs studied include Erd\"{o}s R\'{e}nyi, 1-D geometric, and Barab\'{a}si-Albert random graphs (\cite{Zhang2012robustness,Zhang2012a,Zhang2015a}); random intersection graphs (\cite{Zhao2017robust}); and random interdependent networks (\cite{Shahrivar2015a,Shahrivar2017}). A different approach taken by several authors has been creating methods to systematically construct graphs with guaranteed $r$-robustness or $(r,s)$-robustness (see \cite{LeBlanc_2013_Res, Zhang2012robustness, Guerrero2016formations, Saldana2016triangular}). The authors of \cite{Saldana2016triangular} introduced a method for constructing undirected graphs of arbitrary size that are $(2,2)$-robust. Another recent work introduced algorithms to construct undirected graphs of arbitrary robustness and to either increase or decrease the robustness of a particular graph whose current robustness is known \cite{Saldana2016triangular}. The method allows for the creation of robust undirected graphs with minimal number of nodes. We point out that a majority of the recent approaches to robustness determination have focused on undirected graphs.

Motivated by this robustness determination problem, we introduce a class of directed graphs in which each node's set of in-neighbors is determined by a connection parameter $k$. We show that graphs of this type have an $r$-robustness and $(r,s)$-robustness that is a function of $k$. Since these graphs have a determined robustness, they can be used for any consensus algorithm that depends on a specified robustness level (e.g. the W-MSR algorithm). Our method allows for the creation of graphs of arbitrary robustness and is scalable with respect to the number of nodes.

The paper is organized as follows: Section \ref{notation} outlines the notation used throughout this paper. Section \ref{robustness} presents our main theorem and results. Section \ref{simulations} presents simulations that support our results. Finally, Section \ref{conclusion} summarizes this paper and outlines potential future work.

\section{Notation, Circulant Graphs, and \emph{r}-Robustness}
\label{notation}

\subsection{Graph Theory Notation}


We denote a digraph as $\mathcal{D} = (\mathcal{V},\Ed)$, with $\V = \{1,...,n\}$ denoting the vertex set, or agent set, of the graph and $\Ed$ denoting the edge set of the graph. The set $\V$ is divided into agents that are behaving normally $\N$, and agents that are misbehaving or adversarial $\mathcal{A}$. A directed edge ${(i,j) \in \Ed}: i,j \in \V$ denotes that there exists a connection from node $i$ to node $j$, but not vice-versa. Agent $j$ is able to receive information from agent $i$ if $(i,j)$ is in $\Ed$. We call agent $i$ an in-neighbor of $j$ and agent $j$ an out-neighbor of $i$. We denote the in-neighbor set of a node $j$ as $K_j = \{i \in \mathcal{V}: (i,j) \in \Ed\}$.

We denote the cardinality of a set $S$ as $|S|$. We denote the set of integers as $\mathbb{Z}$ and the set of integers greater than or equal to 0 as $\mathbb{Z}_{\geq 0}$. In addition, we denote the set of natural numbers as $\mathbb{N}$.

An undirected graph of $n$ nodes is called circulant if there exists a set $\{a_1, a_2, \ldots, a_l \in \mathbb{Z}_{\geq 0}: a_1 < a_2 < \ldots < a_l < n\}$ such that $(i, \left[i \pm a_1 \right] \text{mod}\, n) \in \Eg, \ldots, (i,\left[i \pm a_l \right] \text{mod}\, n) \in \Eg$ \cite{Boesch1984}. We call such a graph an \emph{undirected circulant graph}. It should be noted that these graphs are constructed over the additive group of integers modulo $n$ (the nodes $n+a$ and $a$ are congruent modulo $n$). We now define a similar concept for directed graphs as follows:

\begin{define}
A digraph of $n$ nodes is called \emph{circulant} if there exists a set $\{a_1, a_2, \ldots, a_m: 0 < a_1 < a_2 < \ldots < a_m < n\},\ m \in \mathbb{Z}_{\geq 0}$ such that $(i, \left[i+a_1 \right] \text{mod}\, n) \in \mathcal{E}_d, \ldots, (i,\left[i+a_m \right] \text{mod}\, n) \in \mathcal{E}_d$. We denote such a graph as $C_n(a_1, a_2, \ldots, a_m) = (\V, \E_d)$ and call it a \emph{directed circulant graph} or \emph{circulant digraph}.
\end{define}

We point out that the name \emph{circulant} arises from the fact that the adjacency matrix for such a graph is a circulant matrix; i.e. a matrix where each row is defined by cyclically shifting every entry of the previous row one entry to the right. The matrix can therefore be defined by the entries of its first row (\cite{Boesch1984,Elspas1970}). A network of agents with this communication topology should not be confused with a network of agents in a physically circular formation. We emphasize that as long as the network satisfies the conditions outlined, it can be called circulant regardless of the physical orientation of the agents.

\begin{figure}[htp]
\centering
    \begin{minipage}{0.45\textwidth}
        \centering
        \begin{equation*}
\begin{bmatrix}
b_0 & b_1 & b_2 & \ddots & b_n \\
b_n & b_0 & b_1 & \ddots & b_{n-1} \\
b_{n-1} & b_n & b_0 & \ddots & b_{n-2} \\
\ddots & \ddots & \ddots & \ddots & \ddots \\
b_1 & b_2 & b_3 & \ddots & b_0
\end{bmatrix}
\end{equation*}
\caption{The general structure of a circulant matrix. By defining the first row, the rest of the matrix is determined. Circulant digraphs have circulant adjacency matrices.}
    \end{minipage}\hfill
    \begin{minipage}{0.45\textwidth}
        \centering
        \includegraphics[width=0.45\textwidth]{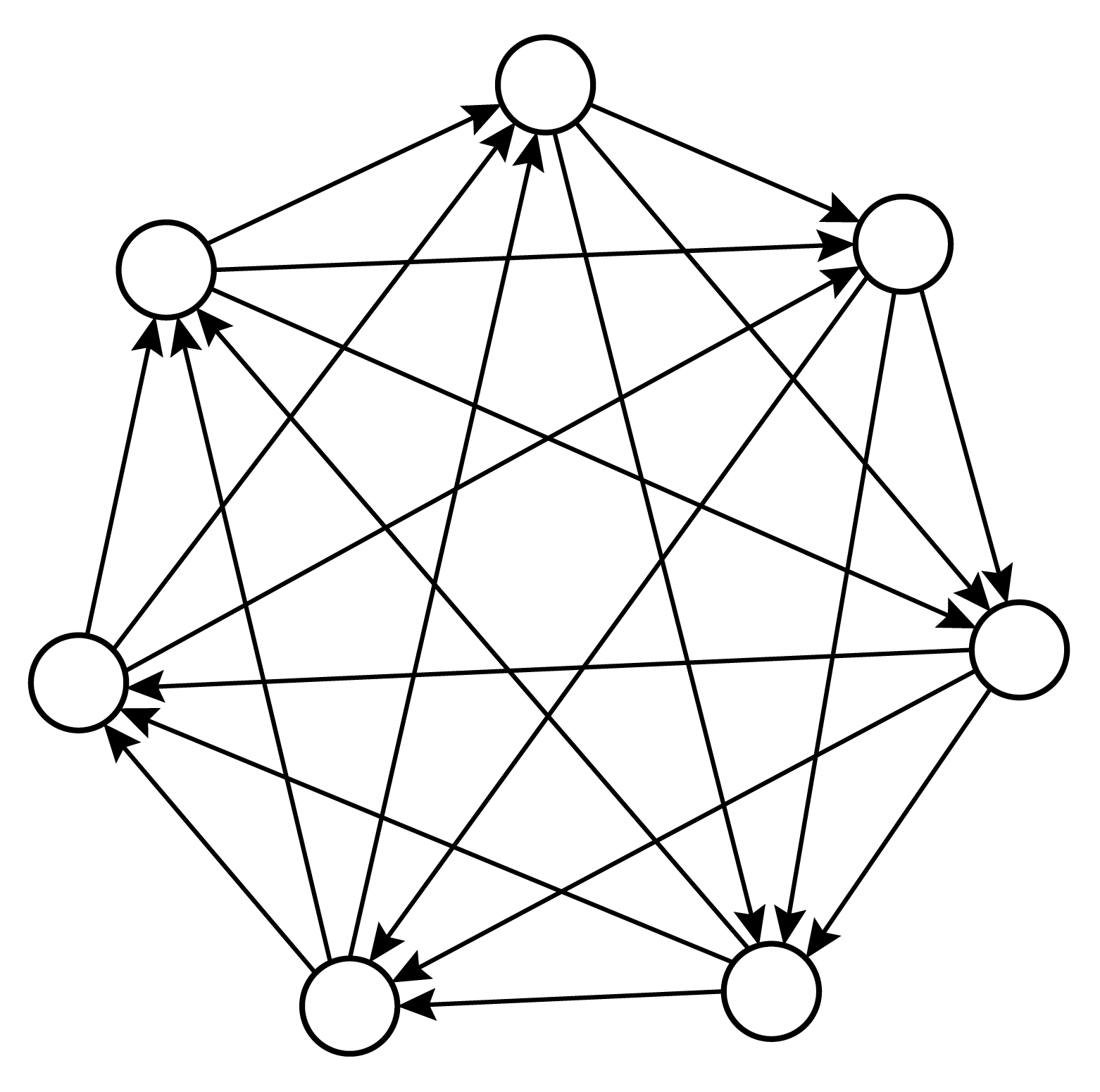}
        \caption{A 3-circulant digraph on 7 nodes, denoted $C_7\{1,2,3\}$.}
        \label{fig:3circ}
    \end{minipage}
\end{figure} 

In this paper, we analyze the robustness properties of a specific class of circulant digraphs which we call \emph{k-circulant digraphs}:

\begin{define}
Let $n \in \mathbb{Z},\ n \geq 2$ and let $k \in \mathbb{Z}: 1 \leq k \leq n-1$. A $k$-circulant digraph  is any circulant digraph of the form $C_n(1,2,3,\ldots,k) = (\V, \E_d)$.
\end{define}

This type of graph is fully determined by the number of nodes $n$ and by the parameter $k$, which determines the in- and out-neighbors of each node. In a graph without self-loops and without more than one edge between any two nodes, $1 \leq k \leq n-1$. When $k=n-1$, the graph becomes a complete graph.

\subsection{Reachability, $r$-Robustness, and $(r,s)$-Robustness}

The notions of reachability, \emph{r}-robustness, and $(r,s)$-robustness were defined by the authors of \cite{LeBlanc_2013_Res}. Although the definitions refer specifically to digraphs, they also apply to undirected graphs.\footnote{Undirected graphs can be modeled as digraphs in which $(i,j) \in \E \iff (j,i) \in \E$}

\begin{define}[\emph{Reachability}]
Consider a digraph $\mathcal{D} = (\V, \Ed)$ and a nonempty subset of nodes $\mathcal{S} \subset \mathcal{V}$. The set $\mathcal{S}$ is called \emph{r-reachable} if $\exists i \in \mathcal{S}$ such that $| K_i \setminus \mathcal{S}| \geq r,\ r \in \mathbb{Z}_{\geq 0}$.
\end{define}

\begin{define}[\emph{r-Robustness}]
A digraph $\mathcal{D}$ is called \emph{r-robust}, with $r \in \mathbb{Z}_{\geq 0}$, if for every nonempty, disjoint pair of subsets $S_1$ and $S_2$ of $\mathcal{V}$, at least one subset is \emph{r}-reachable. In other words, $\forall S_1 \subset \V$ and $\forall S_2 \subset \V$ such that $S_1 \neq \{\emptyset\},\ S_2 \neq \{\emptyset\}$, and $S_1 \cap S_2 = \{\emptyset\}$, $\exists i \in S_1: |K_i \backslash S_1| \geq r$ or $\exists j \in S_2: |K_j \backslash S_2| \geq r$.
\end{define}

\begin{define}[\emph{(r,s)-Robustness}]
Consider a nonempty and nontrivial digraph on $n \geq 2$ nodes, $\mathcal{D} = (\mathcal{V},\mathcal{E}_d)$. Define $r \in \mathbb{Z}_{\geq 0}$ and $s \in \mathbb{Z}: 1 \leq s \leq n$. Also, define the set $\mathcal{X}_{S_m}^r = \{i \in S_m: | K_i \setminus S_m| \geq r\}$ for $m \in \{1,2\}$, where $S_1$ and $S_2$ are nonempty, disjoint subsets of $\V$. Then the digraph $\mathcal{D}$ is called $(r,s)$\emph{-robust} if for every pair of subsets $S_1,\ S_2$ of $\mathcal{V}$, one of the following conditions holds:

\begin{enumerate}
	\item $|\mathcal{X}_{S_1}^r| = |S_1|$
	\item $|\mathcal{X}_{S_2}^r| = |S_2|$
	\item $|\mathcal{X}_{S_1}^r| + |\mathcal{X}_{S_2}^r| \geq s$
\end{enumerate}
\end{define}

\subsection{The W-MSR Algorithm}

The W-MSR algorithm is based upon a linear consensus protocol. The linear consensus protocol operates by updating each normal agent's state $x[t] \in \mathbb{R}$ according to the equation

\[x_i[t+1] = \sum_{j \in K_i \cup \{i\}} w_{ij}[t]x_j[t],\ \forall i \in \N  \]

It is assumed that the following conditions hold for the weights $w_{ij}[t]$ for all $i \in \N$ and for all $t \in \mathbb{Z}_{\geq 0}$:

\begin{itemize}
\item $w_{ij}[t] = 0$ when $j \notin K_i \cup \{i\}$
\item $w_{ij}[t] \geq \alpha,\ 0 < \alpha < 1,\ \forall j \in K_i \cup \{i\}$
\item $\sum_{j=1}^n w_{ij}[t] = 1$
\end{itemize}

The W-MSR algorithm with parameter $F \geq 0$ alters the above protocol by having each agent remove state values that are relatively extreme compared to the rest of the agent's in-neighbor set and its own state. Specifically, it works as follows \cite{LeBlanc_2013_Res}:

\begin{enumerate}
\item At each time step $t$, every normal agent $i \in \N$ forms a sorted list of the state values of its in-neighbors and its own state
\item If there are $F$ or less values greater than its state value, each agent $i$ removes those values from the list. If there are greater than $F$ states greater than its own state, it removes the highest $F$ states. In addition, if there are $F$ or less values less than its state, it removes those values. If there are more than $F$ values less than its state value, it removes the smallest $F$ values from the list. 
\item  The set of in-neighbors for each $i$ whose state values were not removed from the list at $t$ is denoted $\mathcal{P}_i[t]$. Using the state values remaining in the list, each agent $i$ updates its state as follows:
\begin{equation}
\label{inputeq}
x_i[t+1] = \sum_{j \in \mathcal{P}_i \cup \{i\}} w_{ij}[t]x_j[t]
\end{equation}
\end{enumerate}

The main advantage of the W-MSR algorithm and related algorithms is that they allow normally behaving nodes in a network with only local information to achieve consensus in the presence of misbehaving nodes. The algorithm allows normal nodes to achieve consensus without any global knowledge of the network structure and without the need for any node to identify the misbehaving agents in its in-neighbor set. The measures of $r$- and $(r,s)$-robustness describe the amount of misbehaving agents a network can tolerate.\footnote{Misbehaving agents refer to agents that do not update their states according to the nominal state update protocol. A malicious agent is defined as an agent $i$ who does not apply the W-MSR protocol to update its state $x_i[t]$, but at each time step sends the same value $x_i[t]$ to all its out-neighbors. A Byzantine agent is defined as an agent which at each time step either sends different values to different out-neighbors or does not apply the W-MSR protocol to update its state (see \cite{LeBlanc_2013_Res,LeBlanc_2013_Res_Continuous,Zhang2012c})}. A network being $(2F+1)$-robust is a sufficient condition for the normal nodes using the W-MSR algorithm to achieve consensus to a value within the convex set of the maximum and minimum initial states if each node has no more than $F$ malicious in-neighbors. It is also a sufficient condition for the normal nodes to achieve this same kind of consensus if there are $F$ total Byzantine nodes (and no other misbehaving nodes) in the network. A network being $(F+1,F+1)$-robust is a necessary and sufficient condition for the normal nodes to achieve consensus when no more than $F$ total malicious nodes are present in the entire network \cite{LeBlanc_2013_Res}.  

\section{Robustness of Circulant Graphs}
\label{robustness}

\subsection{Undirected Circulant Graphs}

Theorem 4 of \cite{Zhang2015a} demonstrates that if an undirected line or ring graph is $2p$-connected, then it is at least $\floor{\frac{p}{2}}$-robust. This result applies to undirected circulant graphs, which fall under the category of undirected ring graphs. In \cite{LeBlanc_2012_thesis} it is shown that $(r+s-1)$-robustness implies $(r,s)$-robustness. It can then be shown that a $2p$ connected graph is at least $(\floor{\frac{p+2}{2}},\floor{\frac{p+2}{2}})$-robust for even $p$ and at least $(\floor{\frac{p+1}{2}},\floor{\frac{p+1}{2}})$-robust for odd $p$.

However, this theorem does not apply to directed graphs because some ambiguity arises with the definition of vertex connectivity for digraphs. The vertex connectivity of an undirected graph is traditionally defined as the minimum number of vertices whose removal results in either a disconnected graph or a trivial single vertex graph \cite{Rahman2017}.\footnote{An undirected graph is disconnected if there exists two nodes with no path between them.} The authors of \cite{Leblanc2012low} define a digraph to be disconnected if its underlying graph is disconnected, where the underlying graph is the graph created by replacing all directed edges of the graph with undirected edges. This definition of connectivity for a digraph therefore hinges upon the connectivity of its underlying graph. However, using this definition of connectivity it can be shown that there exist digraphs which are not $\floor{\frac{p}{2}}$-robust, but whose underlying graphs are $p$-connected. For example, under the definition just described the graph in Figure \ref{fig:CountEx} would be 4 connected, but is only 1-robust. Another measure of connectivity generalized to digraphs exists, called a minimum vertex disconnecting set (\cite{Tindell1996}; see also the equivalent definition of cutset in \cite{Hamidoune1984}). It can be shown however that graphs with arbitrarily large minimum vertex disconnecting sets are trivially 1-robust, and therefore this metric cannot be used to determine robustness.  A different proof is therefore necessary to demonstrate the robustness of circulant digraphs.


\begin{figure}
\includegraphics[width=\columnwidth]{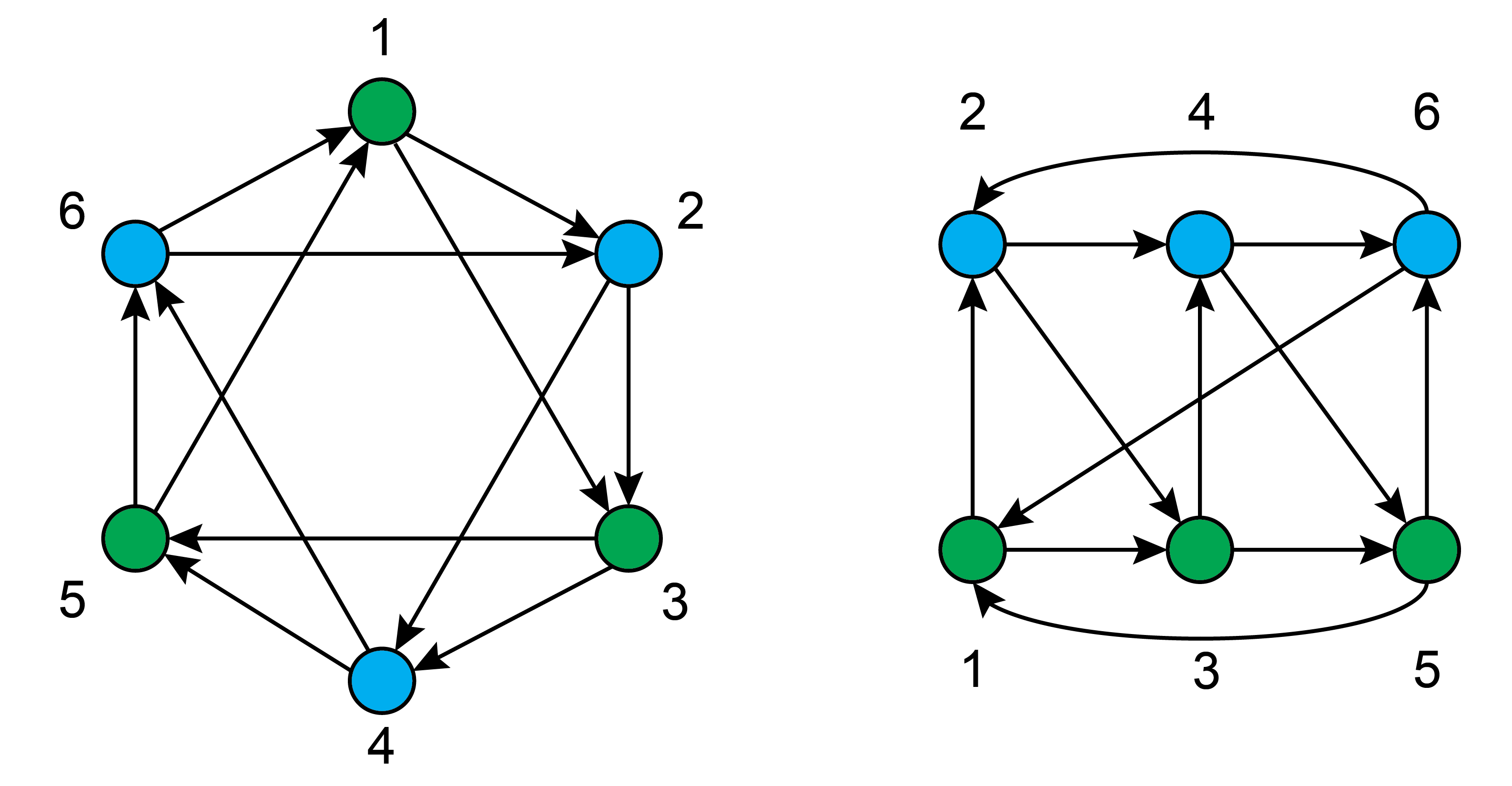}
\caption{Example of a directed graph whose underlying graph is $p$-connected, but which is not $\floor{\frac{p}{2}}$-robust. The graph shown has an underlying graph with vertex connectivity equal to 4. If the nodes of the graph are divided into the two nonempty, disjoint sets denoted by the green and blue colors, each node clearly only has one in-neighbor outside its set. This implies that the graph can be no more than 1-robust. Note that the two arrangements are the exact same graph; the second configuration is rearranged for clarity.}
\label{fig:CountEx}
\end{figure}

\subsection{$r$-Robustness of Circulant Digraphs}

In this section we present our main results, which demonstrate the robustness of k-circulant digraphs. We first establish their $r$-robustness:

\begin{theorem}
\label{circdigraph}
The circulant digraph $C_n(1, \ldots,k) = (\V, \Ed)$ is $\ceil{\frac{k}{2}}$-robust, where $k < n$. Moreover, if $k = n-1$ the graph is $\ceil{\frac{n}{2}}$-robust.
\end{theorem}

\begin{proof}
The scenario that limits robustness is when the sets $S_1$ and $S_2$ allow for the least number of in-neighbors outside of any agent's set \cite{Guerrero2016formations}. This is the case when $S_1 \cup S_2 = \V$, with $S_1 \cap S_2 = \{\emptyset\}$, and therefore we proceed with this assumption.

From the definition of $r$-robustness it follows that if a graph is \emph{not} $\sigma$-robust for some value $\sigma \in \mathbb{Z}$, then $\exists S_1,S_2$ such that $\forall i \in S_1,\ |K_i \backslash S_1| < \sigma$ and $\forall j \in S_2,\ |K_j \backslash S_2| < \sigma$. To prevent confusion, we clarify that the phrase $\sigma$-robust simply means the graph is $r$-robust for the value $r = \sigma$. No graph can be less than $0$-robust, hence $\sigma \geq 0$.

Suppose that a k-circulant graph is not $\sigma$-robust. Without loss of generality, this implies that there are $S_1$ and $S_2$ such that for any node $i \in S_1$ there exists a $b \in \mathbb{Z},\ 0 < b < \sigma$ such that node $i+b \in S_2$ and nodes $\{i+1, \ldots, i+b-1\} \in S_1$. This can be seen by noting that $S_2$ is nonempty, and if $b \geq \sigma$ then $|K_{i+b}\backslash S_2| \geq \sigma$, contradicting our initial assumption.

Next, note that the in-neighbor set of $i$ is $K_i = \{i-k, \ldots, i-1\}$ and the in-neighbor set of $i+b$ is $\K{i+b} = \{i+b-k, \ldots, i+b-1\}$. The intersection of these two in-neighbors sets is $K_i \cap K_{i+b} = \{i+b-k, \ldots, i-1\}$ Denote the number of $S_1$ nodes and the number of $S_2$ nodes in the set $\{i-k, \ldots, i+b-k-1\}$ as $\alpha_1$ and $\beta_1$ respectively, with $\alpha_1,\beta_1 \in \mathbb{Z}_{\geq 0}$. Denote the number of $S_1$ nodes and $S_2$ nodes in the set $\{i+b-k, \ldots, i-1\}$ as $\alpha_2$ and $\beta_2$ respectively, with $\alpha_2, \beta_2 \in \mathbb{Z}_{\geq 0}$.

Observe that $|K_i \backslash S_1| = \beta_1 + \beta_2$. Also note that ${|K_{i+b} \backslash S_2|} = \alpha_2 + b$ (since nodes $\{i, \ldots, i+b-1\} \in S_1$). By our robustness assumption, we must have that $\beta_1+\beta_2 \leq \sigma-1$ and $\alpha + b \leq \sigma-1$. The definition of a k-circulant graph also implies that $|K_{i+b}| =  \alpha_2 + \beta_2 + b = k$. From these equations we obtain
\begin{align*}
\alpha_2 + b + \beta_2 &\leq \sigma -1 + \beta_2\\
k &\leq \sigma - 1 + \beta_2\\
k-\sigma+1 &\leq \beta_2
\end{align*}
This then implies that $\beta_1 + k - \sigma + 1 \leq \beta_1 + \beta_2 \leq \sigma-1$. Rearranging we obtain $\beta_1 \leq 2 \sigma - (k+2)$.  Since $0 \leq \beta_1$ we then have
\begin{align*}
0 &\leq 2 \sigma - (k+2) \\
\frac{k}{2} + 1 &\leq \sigma
\end{align*}
Since $\sigma \in \mathbb{Z}$, this implies that the smallest value of $\sigma$ for which a k-circulant graph is \emph{not} $\sigma$-robust is $\frac{k}{2} + 1$ for even $k$ and $\ceil{\frac{k}{2}+1}$ for odd $k$. Therefore a $k$-circulant graph must be $\ceil{\frac{k}{2}}$-robust.

Lastly, the case when $k = n-1$ implies a complete graph. From \cite{LeBlanc_2013_Res} it can be shown that such graphs are $\ceil{\frac{n}{2}}$-robust.
\end{proof}

Since the robustness is a function of $k$ only and not of $n$, it is worth noting that the $r$-robustness of circulant digraphs can be determined regardless of the size of the network. As a result these graphs can easily be scaled to any number of nodes while maintaining a desired robustness level. The main limitation is that $k \leq n-1$, implying that a graph with a desired robustness will require a minimum number of nodes.

\begin{figure}
\includegraphics[width=\columnwidth]{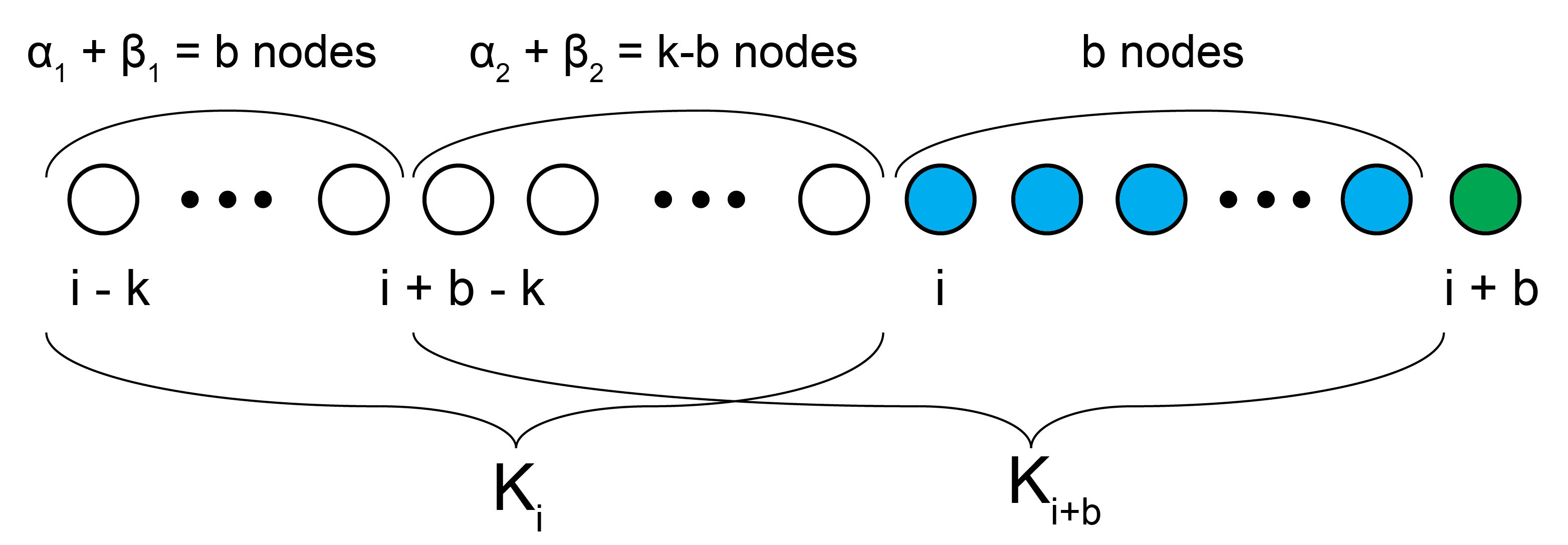}
\caption{Visualization of the sets $K_i$ and $K_{i+b}$, and the values $\alpha_1,\ \alpha_2,\ \beta_1,\ \beta_2$. Here, $i \in S_1$ with $S_1$ represented by the color blue. From the proof, there exists a node $i+b \in S_2$, with $S_2$ represented by the color green. Nodes $i-k$ through $i-1$ are either in $S_1$ or $S_2$.}
\end{figure}

\subsection{$(r,s)$-Robustness of Circulant Digraphs}

In \cite{LeBlanc_2012_thesis} a connection between $r$-robustness and $(r,s)$-robustness was given:

\begin{lemma}
\label{r2rs}
If $\mathcal{D} = (\V, \Ed)$ is $(r+s-1)$-robust with $r \in \mathbb{Z}_{\geq 0}$, $s \in \mathbb{N}$, and $1 \leq r+s-1 \leq \ceil{\frac{n}{2}}$, then $\mathcal{D}$ is $(r,s)$-robust.
\end{lemma}

The proof is outlined in \cite{LeBlanc_2012_thesis}. It should be noted that this is a sufficient condition only, and so the graph may actually have a higher $(r,s)$-robustness (e.g. consider a complete graph). We use this lemma to demonstrate the relationship between a lower bound of $(r,s)$-robustness of $C_n(1, \ldots,k)$-circulant digraphs and the parameter $k$:

\begin{cor}
\label{cor:rsrob}
The circulant digraph $C_n(1, \ldots,k) = (\V, \Ed)$ is at least $(\floor{\frac{k+2}{4}},\floor{\frac{k+2}{4}})$-robust if $k$ is even and at least $(\floor{\frac{k+3}{4}},\floor{\frac{k+3}{4}})$-robust if $k$ is odd.
\end{cor}

\begin{proof}
If $k$ is even, then $C_n(1, \ldots,k)$ is at least $\frac{k}{2}$-robust by Theorem \ref{circdigraph}. Since we are interested in establishing an upper bound $F$ on the number of adversaries in the network, we seek to find the maximum value of $F$ for which the network is $(F+1, F+1)$-robust. This implies $r = s$ for the network's $(r,s)$-robustness. Hence by Lemma \ref{r2rs}, and letting $r = s$,
\begin{align*}
r+s-1 &= \frac{k}{2}  \\
r+s &= \frac{k+2}{2} \\
r = s &= \frac{k+2}{4} \\
&\geq \floor*{\frac{k+2}{4}}
\end{align*}
If $k$ is odd, then $C_n(1, \ldots,k)$ is at least $\ceil*{\frac{k}{2}}$-robust by Theorem \ref{circdigraph}. Hence
\begin{align*}
r+s-1 &= \ceil*{\frac{k}{2}} = \frac{k+1}{2} \\
r+s &= \frac{k+3}{2} \\
r = s &= \frac{k+3}{4} \\
&\geq \floor*{\frac{k+3}{4}}
\end{align*}
\end{proof}



\section{Simulations and Discussion}
\label{simulations}


To demonstrate the robustness of these graphs, we present simulations of agents in a k-circulant network running the W-MSR protocol. The network size is $n=15$ nodes. Each agent in the graph has state $x[t] \in \mathbb{R}$, and each normal agent follows the W-MSR algorithm to update its own state at each time step. The initial state value for each agent is a random value on the interval $[-50, 50]$.

Several models exist to describe the number and distribution of misbehaving nodes in a network, including the $F$-total, $F$-local, and $f$-fraction local models (see \cite{LeBlanc2012a,LeBlanc_2013_Res_Continuous}). For our simulations we consider an $F$-local model, meaning that any normal agent has at most $F$ misbehaving agents in the set of its in-neighbors. Theorem 2 and Corollary 4 of \cite{LeBlanc_2013_Res} establish that $(2F+1)$-robustness is a sufficient condition for a network using the W-MSR algorithm to achieve consensus among its normal nodes under an $F$-local model of misbehaving agents.

We consider two graphs on 15 nodes, each with different values of $k$. The first graph $\D_1$ has $k=6$, implying $\D_1 = C_{15}(1,2,\ldots,6)$. Figure \ref{fig:topo} shows the communication topology of $\D_1$. By Theorem \ref{circdigraph} and Corollary \ref{cor:rsrob}, $\D_1$ is $3$-robust, implying that consensus is guaranteed under a $F$-local malicious adversary model with $F=1$. Figure \ref{fig:Sim1} shows our simulation with $F=1$ and with nodes 1 and 7 misbehaving. Note that any normal agent $i$ has at most 1 misbehaving agent in $K_i$.The red dotted lines represent the state values of misbehaving nodes, while the solid colored lines represent the state values of normal nodes. The normal nodes are clearly able to achieve consensus in the presence of the misbehaving nodes.

\begin{figure}
\includegraphics[width=\columnwidth]{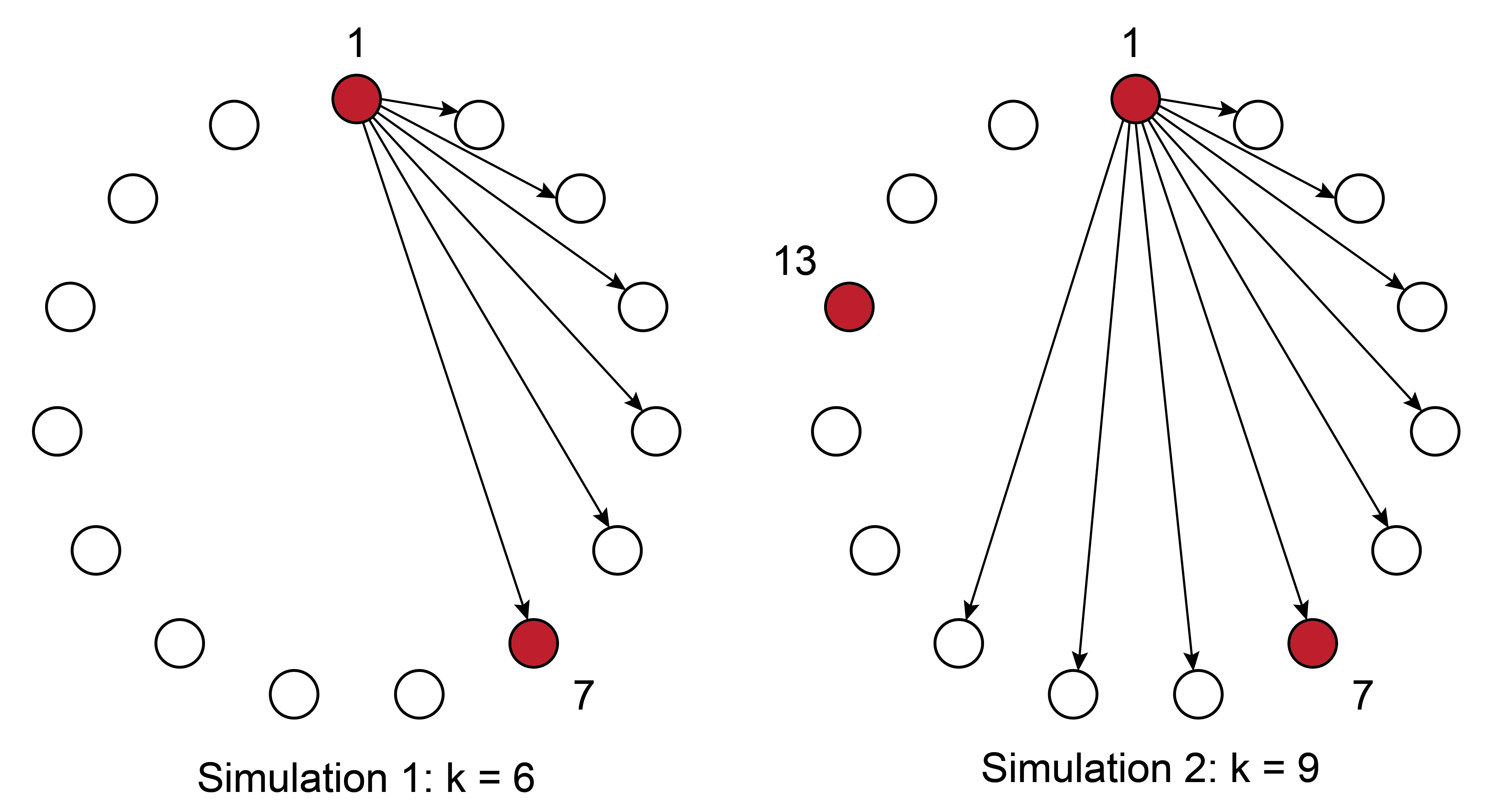}
\caption{The network topology of digraphs $D_1$ and $D_2$. For sake of clarity, only the edges extending from one node are shown; in the actual graph, each node has the same pattern of edges extending from it. The first graph simulated is a $C_{15}\{1,\ldots,6\}$ circulant digraph. The second is a $C_{15}\{1,\ldots,9\}$ circulant digraph. In the first graph, nodes 1 and 7 are misbehaving. In the second, nodes 1, 7, and 13 are misbehaving. The nodes are visualized in a circular manner for ease of understanding rather than representing any kind of physical arrangement.}
\label{fig:topo}
\end{figure}

\begin{figure}
\includegraphics[width=\columnwidth]{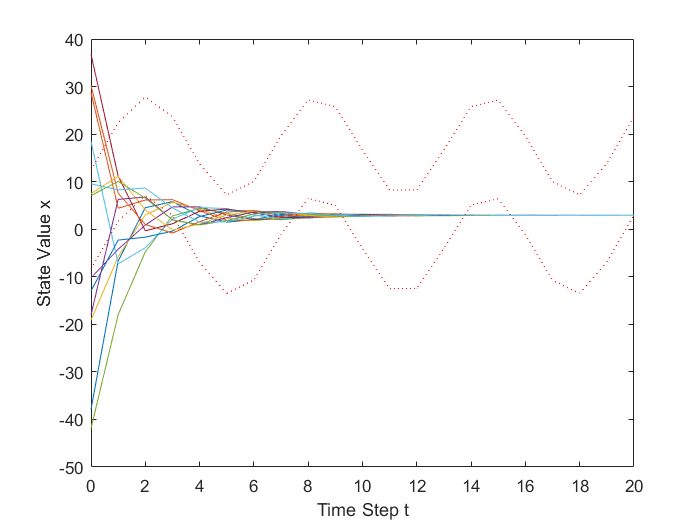}
\caption{Simulation on the graph $\D_1 = C_{15}(1,2,\ldots,6)$. The dotted red lines represent the state trajectories of the misbehaving agents.}
\label{fig:Sim1}
\end{figure}

The second graph $\D_2$ has $k=9$, and therefore is $5$-robust which guarantees consensus under an $F$-local malicious adversary model with $F=2$. Agents 1, 7, and 13 are misbehaving, which implies that any normal agent $i$ has at most 2 misbehaving agents in its in-neighbor set $K_i$. Figure \ref{fig:Sim2} shows the simulation results for the second graph with $F=2$. Again, the normal agents are clearly able to achieve consensus in the presence of the misbehaving nodes. This second simulation also demonstrates the simplicity of changing the robustness of $k$-circulant digraphs. By varying $k$, the network's robustness can be increased or decreased to a desired level.

\begin{figure}
\includegraphics[width=\columnwidth]{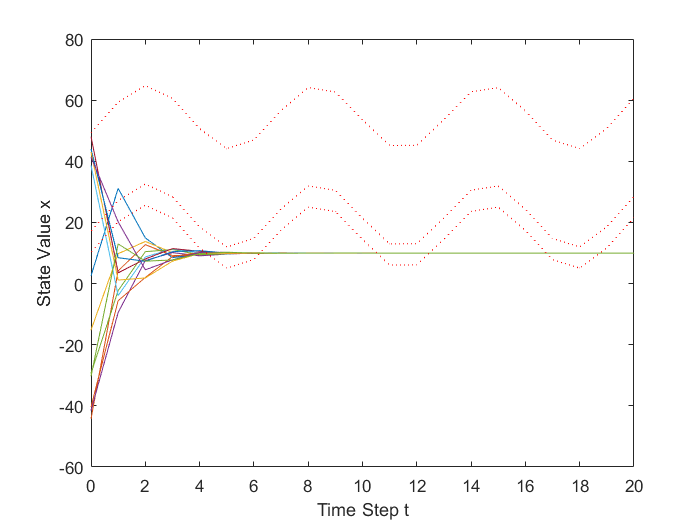}
\caption{Simulation on the graph $D_2 = C_{15}(1,2,\ldots,9)$.}
\label{fig:Sim2}
\end{figure}


\section{Conclusion and Future Work}
\label{conclusion}

This paper demonstrated that a class of scalable graphs called $k$-circulant digraphs with a connection parameter $k$ have $r$-robustness and $(r,s)$-robustness properties that are functions of $k$. Future work includes seeking additional classes of graphs that have predetermined robustness properties, and implementing these classes of graphs in settings that require resilient consensus.

\bibliographystyle{IEEEtran}

\bibliography{Mendeley.bib}
\end{document}